%% file: metric-distortion_arxiv.tex
\documentclass{article}

\usepackage{graphicx}
\usepackage{natbib}

\usepackage{amsmath,amssymb,mathtools,enumerate,amsthm}
\usepackage{paralist}
\input{mydefs.tex}
 
\begin{document}

\title{Metric Distortion of Social Choice Rules: Lower Bounds and Fairness Properties}
\author{Ashish Goel\footnote{Management Science and Engineering, Stanford University}, Anilesh K. Krishnaswamy\footnote{Electrical Engineering, Stanford University} and Kamesh Munagala\footnote{Computer Science, Duke University}
}
\maketitle
\begin{abstract}
 We study social choice rules under the utilitarian distortion framework, with an additional metric assumption on the agents' costs over the alternatives. In this approach, these costs are given by an underlying metric on the set of all agents plus alternatives. Social choice rules have access to only the ordinal preferences of agents but not the latent cardinal costs that induce them. Distortion is then defined as the ratio between the social cost (typically the sum of agent costs) of the alternative chosen by the mechanism at hand, and that of the optimal alternative chosen by an omniscient algorithm. The worst-case distortion of a social choice rule is, therefore, a measure of how close it always gets to the optimal alternative without any knowledge of the underlying costs. Under this model, it has been conjectured that Ranked Pairs, the well-known weighted-tournament rule, achieves a distortion of at most 3 \citep{anshelevich2015approximating}. We disprove this conjecture by constructing a sequence of instances which shows that the worst-case distortion of Ranked Pairs is at least 5. Our lower bound on the worst case distortion of Ranked Pairs matches a previously known upper bound for the Copeland rule, proving that in the worst case, the simpler Copeland rule is at least as good as Ranked Pairs. And as long as we are limited to (weighted or unweighted) tournament rules, we demonstrate that randomization cannot help achieve an expected worst-case distortion of less than 3. Using the concept of approximate majorization within the distortion framework, we prove that Copeland and Randomized Dictatorship achieve low constant factor fairness-ratios (5 and 3 respectively), which is a considerable generalization of similar results for the sum of costs and single largest cost objectives. In addition to all of the above, we outline several interesting directions for further research in this space.
\end{abstract}

\section{Introduction}\label{sec:intro}
Social choice theory is the science of aggregating the varied preferences of multiple agents into a single collective decision. Ways of doing this aggregation are called social choice rules -- functions that map the given preferences of agents, typically in the form of total orderings over a set of alternatives, to a single alternative. The conventional approach to reasoning about the quality of outcomes obtained from these rules has been a normative, axiomatic one. A variety of axiomatic criteria, corresponding to naturally desirable properties, have been proposed, and a great deal of work has been done to understand which axioms can or cannot be satisfied together, and how the known social choice rules measure up against them. For instance, a few celebrated results \citep{gibbard1973manipulation,satterthwaite1975strategy} rule out the concurrent satisfiability of such basic axioms, and additional spatial assumptions that help sidestep these impossibilities have been identified \citep{moulin1980strategy,barbera2001introduction}.

Another approach, which has received a great deal of attention lately \citep{procaccia2006distortion,caragiannis2011voting,boutilier2015optimal}, is to assume a \emph{utilitarian} view, as is commonplace in economics and algorithmic mechanism design. Every agent has latent cardinal preferences over the alternatives in terms of utility (or cost) and the social utility of an alternative is a function of the agents' utilities. The most commonly used objective is the total sum of agent utilities for an alternative. Social choice rules are then viewed as approximation algorithms which try to choose the best alternative given access only to ordinal preferences. Similar to the competitive ratio of online approximation algorithms, the quantity of interest here is the worst-case value (over all possible underlying utilities) of the \emph{distortion} -- the ratio  of the social utility of the truly optimal alternative over that of the alternative chosen by the social choice rule at hand \citep{procaccia2006distortion}. 

Characterizing the worst-case distortion of social choice functions has recently emerged as the central question within the utilitarian approach to social choice. Without any assumptions on the utilities, the distortion of deterministic social choice rules is unbounded \citep{procaccia2006distortion}, and that of randomized social choice rules is $\Omega(m)$, where $m$ is the number of alternatives \citep{boutilier2015optimal}.

Interestingly, some constant factor bounds on the distortion of social choice rules are made possible with an additional \emph{metric} assumption on the cardinal preferences of agents, represented by their costs with respect to alternatives \citep{anshelevich2015approximating}. These costs are assumed to form an unknown, arbitrary metric space, and distortion is redefined in terms of these costs. In this setting, the best known positive result for deterministic social rules is that the distortion of the Copeland rule, a tournament function, is at most $5$ \citep{anshelevich2015approximating}. It is known that the worst-case distortion of {\em any} deterministic social choice rule is at least $3$, and that Ranked Pairs, a weighted-tournament function, achieves this lower bound given some assumptions on the ordinal preferences of agents. It is also known that the distortion of {\em any} randomized rule is at least $2$, and that of Randomized Dictatorship is at most $3$ \citep{anshelevich2016randomized}, showing that randomization helps beat the performance of common deterministic rules with respect to worst-case distortion.

An important open question here is whether the worst-case distortion of Ranked Pairs is indeed $3$, as has been conjectured \citep{anshelevich2015approximating}. And given that tournament (or weighted-tournament rules) provide the best known bounds in the deterministic case, another interesting question is whether they perform well in the randomized case too, and in particular, better than Randomized Dictatorship. The first half of this paper is devoted to settling these questions, providing answers in the negative for both. The proof for Ranked Pairs is particularly surprising and
intricate.

Under the utilitarian metric distortion approach, in addition to
  reasoning about the total social cost, it is natural to ask how ``fair" choosing a particular alternative
  is in terms of the cost incurred by the agents. For example, let us
  say there are two agents and two alternatives, and the costs
  incurred by the agents are $x_1,x_2$ for the first alternative, and
  $(x_1+x_2)/2,(x_1+x_2)/2$ for the second. It seems natural that the
  second is more ``fair" than the first. Various notions of fairness
  such as lexicographic fairness\footnote{This is also called max-min
    fairness in the resource allocation literature.}, prefix-based
  measures, and the more general approximate majorization measure,
  have been studied in the context of routing, bandwidth allocation
  and load balancing problems \citep{kleinberg1999fairness,kumar2000fairness,goel2001approximate}. Other measures of fairness such as envy-freeness \citep{chen2013truth,caragiannis2009low}, maximin-shares \citep{procaccia2014fair}, and leximin \citep{barbara1988maximin,rawls2009theory} have also been studied in the context of mechanism design and social choice. We could also look at objectives given by $\alpha$-percentiles ($\alpha = 0.5$ is
  the median) for $\alpha \in [0,1]$. While there are positive results
  (as achieved by Copeland) for $\alpha \in [0.5,1)$, it is known that
  for $\alpha \in [0,0.5)$, all deterministic social choice rules have
  unbounded worst-case distortion
  \citep{anshelevich2015approximating}.
 
  A question arises as to which of the above notions of fairness can
  be adapted to the metric distortion framework, and moreover yields
  meaningful results. Even if percentiles were the the most appropriate, they are too strong for this domain given the lower bounds \citep{anshelevich2015approximating} mentioned above; in fact, there are very few
  resource allocation settings known where simultaneous maximization
  or approximation of all percentiles is possible. Notions from
  cake-cutting such as envy-freeness and maximin shares do not apply
  to social choice settings, since these definitions assume a
  partition of available goods among agents; in social choice
  settings, we make a single societal decision rather than partition
  resources. We will instead look at approximate majorization, which
  attempts to minimize, simultaneously for all $k$, the sum of the $k$-largest costs incurred by agents. This generalizes both
  lexicographic fairness and total cost minimization, and applies to
  any setting where agents receive utilities (or costs), regardless of whether
  the underlying problem is one of cake-cutting, resource
  allocation, or social choice. Further, an
  approximate majorization ratio of $\alpha$ guarantees an
  $\alpha$-approximation of a large class of fairness measures,
  including all $p$-moments of costs for agents, for $p \ge 1$. The
  second half of this paper states and proves these fairness results. In particular, we show how the simple Copeland rule approximates a broad class of convex cost functions.

\begin{figure}
  \centering
  \def\svgwidth{\columnwidth}
  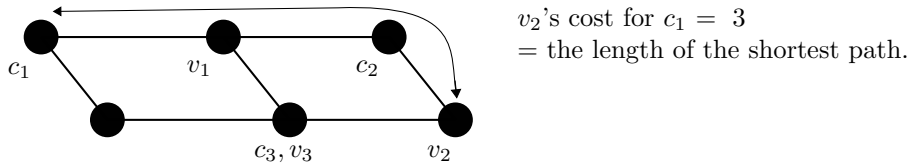
  \caption{Underlying shortest path metric in Example 1 -- each edge is of unit weight}
  \label{fig:graphmetric}
\end{figure}

\paragraph{Warm-up example of metric costs and distortion}
Imagine there are three alternatives $c_1,c_2,c_3$, and three agents $v_1,v_2,v_3$ with preferences $c_1 \succ c_2 \succ c_3$, $c_2 \succ c_3 \succ c_1$ and $c_3 \succ  c_1 \succ c_2$ respectively. The underlying costs are given by the shortest path metric on the graph in Figure \ref{fig:graphmetric}. The agent costs are given by $d(.,.)$ as follows:

\begin{enumerate}[(i)]
\item $d(v_1,c_1) = d(v_1,c_2) = d(v_1,c_3) = 1$; 
\item $d(v_2,c_2) = d(v_2,c_3) = 1$, $d(v_2,c_1) = 3$; 
\item $d(v_3,c_3) = 0$, $d(v_3,c_1) = d(v_3,c_2) = 2$.
\end{enumerate} 

 Let's say that a deterministic rule chooses $c_1$ as the winner based on the preferences. Then the distortion is given by 
${\sum_{i=1,2,3}d(v_i,c_1)}/{\sum_{i=1,2,3}d(v_i,c_3)} = {6}/{2} = 3$. In fact, this metric achieves the worst-case distortion among all possible metrics that agree with the given preferences. 

If a randomized rule picks each of the three alternatives with equal probability, the expected distortion will be equal to $(\frac{1}{3}.6 + \frac{1}{3}.4 + \frac{1}{3}.2) / 2 = 2$. It can be seen, based on symmetry, that this distribution minimizes the worst-case distortion over all possible metrics that agree with the preferences.

Let us also look at the fairness ratio when we pick $c_1$ as the winner. The largest cost for $c_1$ is that of $v_2$ with $d(v_2,c_1) = 3$. And for the optimal alternative $c_3$, the largest cost is $d(v_1,c_3) = 1$. The ratio of these values is $3$. Similarly if we look at the ratio of the two largest costs, then we have $[d(v_2,c_1)+d(v_3,c_1)]/[d(v_1,c_3)+d(v_2,c_3)] = 5/2$. For the three largest costs, we have a ratio of $3$ from above. Therefore, the alternative $c_1$ achieves a fairness ratio of $\max\{3,2.5,3\} = 3$.

\subsection{Our contributions}
Our first set of results are in the negative: we show that social choice rules of simple forms cannot have worst-case distortion ratios matching the known lower bounds. Our second set of results concern defining fairness in this setting, and upper bounding the fairness ratios of natural social choice rules, in particular the Copeland rule.

\paragraph{Lower Bounds on Distortion}
It has been conjectured that the simple Ranked Pairs rule achieves the optimal distortion ratio of $3$ \citep{anshelevich2015approximating}. This conjecture is based on the fact that if the ordinal preferences of agents are restricted to be of a certain form, Ranked Pairs does indeed have a distortion at most $3$ (see Theorem \ref{thm:anshel}). Our first main result is disproving this conjecture -- we show that Ranked Pairs, and the related Schulze rule, have a worst-case distortion ratio at least $5$, and in that sense are no better than the Copeland rule when the preferences are general. We do this by constructing a sequence of instances where the agent preferences are obtained by coupling cyclic permutations of two equally large sets of alternatives in a one to one fashion.

\begin{result}[Theorem \ref{thm:rpbound}]
Ranked Pairs, and the Schulze rule, have a worst-case distortion ratio of at least 5.
\end{result}

As stated before, a lower bound of $2$ is known on the worst-case distortion of any randomized rule \citep{anshelevich2016randomized}. We show that this lower bound cannot be achieved by {\em any} rule that looks at only the pairwise wins/losses among the alternatives, or the margins of these wins/losses (tournament rules and weighted-tournament rules -- see Section \ref{subsec:tourn-rules} for a definition).
\begin{result}[Theorem \ref{thm:rand-tourney}]
The worst-case (expected) distortion of any tournament or weighted-tournament rule is at least 3. 
\end{result}


\paragraph{Fairness properties}
We introduce a method of quantifying the ``fairness'' of social choice rules by incorporating the concept of approximate majorization \citep{goel2006simultaneous} within the \emph{metric distortion} framework. For this purpose, we redefine the social cost of any alternative as the sum of its $k$ largest agent costs. How fair a given social choice rule is depends on how the alternative it chooses performs on this objective compared to every other alternative. To evaluate the fairness of a social choice rule, we then seek to bound the distortion ratio of this objective \emph{simultaneously} over all possible values of $k$:  we call this the \emph{fairness ratio} (we define this measure formally in the next section). 

The fairness ratio generalizes both the sum of costs objective (utilitarianism), and maxmin fairness (egalitarianism). Given such a strong definition, it is impossible to achieve a constant fairness ratio in many settings, and surprisingly, for the metric distortion problem we study in this paper, simple social choice rules like Copeland and Randomized Dictatorship achieve small constant fairness ratios that match the best known distortion bounds for just the sum objective.

\begin{result}[Theorem \ref{thm:copeland-fairness}]
Copeland rule achieves a fairness ratio of at most 5.
\end{result}

\begin{result}[Theorem \ref{thm:rand-dict-fairness}]
Randomized Dictatorship achieves a fairness ratio of at most 3.
\end{result}

Additionally, for deterministic rules, a bound on the fairness ratio translates to an approximation result for a general class of symmetric convex objectives (see Section \ref{subsec:prelims-fairness}). And therefore, the above result leads us to the surprising observation that, assuming metric costs, the simple Copeland rule simultaneously approximates a very broad class of cost functions.

\paragraph{Conjectures and open problems} There are many directions for further research on the metric distortion problem. We mention some of these in context as we go along (Conjectures \ref{conj:det-tourney}, \ref{conj:detrandbound}). More details can be found in the Appendix in Sections \ref{subsec:inst-opt-dist} and \ref{subsec:metriccandidate}. For example, for randomized rules, we suggest an interesting variation of the measure of distortion, one that is stronger adversarially than the standard measure (Section \ref{subsec:metriccandidate}).

\subsection{Related Literature}
Several interesting problems pertaining to the distortion arising from the mapping of cardinal preferences to ordinal information have been studied \citep{moulin2016handbook}. The worst case distortion of social choice rules, with unrestricted or normalized utilities, is known to be unbounded \citep{procaccia2006distortion}. With randomized mechanisms, it is possible to achieve a distortion of $\Omega(\sqrt{m} \log^* m)$, where $m$ is the number of alternatives \citep{caragiannis2011voting}. The standard assumption here is that agents translate cardinal scores into ordinal preferences in the straightforward way -- the alternative with the $k$-th highest utility is placed in the $k$-th position. If this mapping could be done in any other way, it is possible to construct low distortion embeddings of cardinal preferences into (ordinal) social choice rules like plurality \citep{caragiannis2011voting}. Another interesting result here is that is possible to construct a truthful-in-expectation mechanism whose worst-case distortion is $O(m^{3/4})$ \citep{filos2014truthful}.
  
  Analysis of cardinal preferences under spatial models of proximity has had a long history in social choice \citep{enelow1984spatial,moulin1980strategy}. Such models, especially those with euclidean spaces, have also been commonly studied in the approximation algorithms literature on facility location problems \citep{arya2004local,drezner1995facility}. In these models, the cost of an agent for an alternative is given by the distance between the two.
 As mentioned earlier, our work follows the literature on the analysis of distortion of social choice rules under the assumption that agent costs form an unknown metric space \citep{anshelevich2015approximating,anshelevich2016randomized}. We have already mentioned that several lower and upper bounds for both the sum of costs and percentile objectives are known in this setting. In addition, it is known that a distortion of at most $4$ for the median objective can be achieved by a randomized mechanism that chooses from the uncovered set \citep{anshelevich2016randomized}. 
 
 It is important to note that in the special case of euclidean metrics, it possible to design low distortion mechanisms, with the additional constraint of their being truthful-in-expectation \citep{feldman2016voting}. Additionally, the metric distortion framework has also been used to study other problems such as finding an approximate maximum weight matching with access to only ordinal preferences \citep{anshelevich2016blind}.
 
 In the distortion framework, both the interpersonal comparison of utilities, and the goal of utility maximization, are implicitly assumed to be valid. While the interpersonal comparison of utilities is more meaningful in some contexts than others \citep{boutilier2015optimal}, we take it for granted. 
 
 While the results on the distortion of the sum of costs (or utilities) objective are extremely interesting, minimization of total cost (or maximization of total utility) is not the only imaginable goal of social choice mechanisms. The first step toward other understanding the distortion of other objectives is apparent in the various results on the distortion of the median cost objective \citep{anshelevich2015approximating,anshelevich2016randomized}. We take a further step in this direction by drawing on the notions of fairness that have been studied in the context of of network problems such as bandwidth allocation and load balancing \citep{kleinberg1999fairness,kumar2000fairness,goel2001approximate}.

\section{Preliminaries}
\subsection{Social Choice Rules} Let $\V$ be the set of agents and $\C$ the set of alternatives. We will use $N$ to denote the total number of agents, i.e., $N = |\V|$. Every agent $v \in \V$ has a strict (no ties) preference ordering $\sigma_v$ on $\C$. For any $c,\cpr \in \C$, we will use $c \succ_v \cpr$ to denote the fact that agent $v \in \V$ \emph{prefers} $c$ over $\cpr$ in her ordering $\sigma_v$. Let $\s$ be the set of all possible preference orderings on $\V$. We call a profile of preference orderings $\sigma \in \s^N$ as an \emph{instance}. 

Based on the preferences of agents, we want to determine a single alternative as the winner, or a distribution over the alternatives and pick a winner according to it. A deterministic social choice rule is a function $f: \s^N \to \C$ that maps each instance to an alternative. A randomized social choice rule is a function $g: \s^N \to \Delta(\C)$, where $\Delta(\C)$ is the space of all probability distributions over the set of alternatives $\C$. 

To define the social choice rules that we use in this paper, we need a few additional definitions. An alternative $c$ \emph{pairwise-beats} $\cpr$ if $|\{v \in \V : c \succ_v \cpr\}| \geq \frac{N}{2}$, with ties broken arbitrarily. Given an instance $\sigma$, a complete weighted digraph $G_t(\sigma)$ with $\C$ as the set of nodes, and the weight of any edge $c \to \cpr$ given by $w(c,\cpr) = |\{v \in \V : c \succ_v \cpr \}|$, is called the \emph{weighted-tournament} graph induced by $\sigma$. An unweighted digraph $G_m(\sigma)$ with $\C$ as the set of nodes such that an edge from $c \to \cpr$ exists iff $c$ pairwise beats $\cpr$ is called the \emph{tournament} graph induced by $\sigma$. 

\begin{itemize}
\item \textbf{Ranked Pairs}: Given an instance $\sigma$, sort the edges of the weighted-tournament graph $G_t(\sigma)$ according to the values $w(.,.)$ in some non-decreasing order (breaking ties arbitrarily). Start with a graph $G = (\C,\emptyset)$ and iterate over the edges in the order determined above. At each step, add the edge to $G$ if it does not create a cycle, and discard the edge otherwise. The winning alternative is the source node of the resulting directed acyclic graph.
\item \textbf{Copeland}: Given an instance $\sigma$, define a score for each $c \in \C$ given by $|\{ \cpr \in \C : c \mbox{ pairwise beats } \cpr \}|$. The alternative with the highest score (the largest number of pairwise victories) is chosen to be the winner. In other words, the winning alternative is the node in the tournament graph $G_m(\sigma)$ with the maximum out-degree (breaking ties arbitrarily).
\item \textbf{Randomized Dictatorship}: Choose alternative $c \in \C$ with probability $p(c)$ equal to ${|V_c|}/{N}$ where $V_c = \{ v \in \V : c \succ_v \cpr, \; \forall \cpr \neq c\}$.
\item \textbf{Schulze} \citep{schulze2003new} In the weighted-tournament graph, a path of strength $p$ from alternative $c$ to alternative $\cpr$ is a sequence of candidates $c_1, c_2, \ldots, c_n$ with the following properties:
\begin{inparaenum}[(i)]\item $c_1 = c$ and $c_n = \cpr$,
\item for all $i=1,\ldots ,(n-1)$, $w(c_i,c_{i+1}) \geq w(c_{i+1},c_{i})$, and \item for all $i=1,\ldots ,(n-1)$, $w(c_i,c_{i+1}) \geq p$.
\end{inparaenum}

Let $p(c,\cpr)$ be the strength of the strongest path from $c$ to $\cpr$. If there is no path from $c$ to $\cpr$, then $p(c,\cpr) = 0$.

Define a relation $\succ^\star$ as follows: $\forall c,\cpr$, $c \succ^\star \cpr \iff p(c,\cpr) > p(\cpr,c)$. It can be proven that $\succ^\star$ defines a transitive relation. The alternative (with arbitrary tie-breaking, as there may be many such) $c^\star$, such that $p(c^\star,c) \geq p(c,c^\star)$ for all other alternatives $c$, is chosen as the winner.
\end{itemize}

\subsection{Tournament and weighted-tournament rules} \label{subsec:tourn-rules} Any social choice rule that chooses an alternative, or a distribution over the alternatives, based on just the tournament graph is called a tournament rule \citep{moulin2016handbook}. These are also called C1 functions according to Fishburn's classification \citep{fishburn1977condorcet}. Any rule that is a function of the weighted-tournament graph is a \emph{weighted-tournament} rule, as long as it is not a tournament rule \citep{moulin2016handbook}. According to Fishburn's classification, these rules are also called C2 functions \citep{fishburn1977condorcet}. Such rules do not need knowledge of all the preferences orderings, just the aggregated information in terms of the tournament/weighted-tournament graph. From the above definitions, we see that Ranked Pairs is a deterministic weighted-tournament rule, and Copeland a tournament rule. Randomized Dictatorship is neither a tournament rule nor a weighted-tournament rule, because it needs to know which alternative is first in each ordering.

\subsection{Metric costs} We assume that the agent costs over the alternatives is given by an underlying metric $d$ on $\C \cup \V$. $d(v,c)$ is the cost incurred by agent $v$ when alternative $c$ is chosen as the winning alternative. 
\begin{definition} \label{def:metric}
A function $d : \C \cup \V \times \C \cup \V \to \Rplus$ is a metric iff $\forall x,y,z \in \C \cup \V$, we have the following:
\begin{inparaenum}
\item $d(x,y) \geq 0$,
\item $d(x,x) = 0$,
\item $d(x,y) = d(y,x)$ , and
\item $d(x,z) \leq d(x,y) + d(y,z)$.\label{eqn:tri-ineq}
\end{inparaenum}
\end{definition}

We can do with a much simpler yet equivalent assumption on the agents' costs (see Lemma \ref{thm:tri-eq-quad}). We need to first define a q-metric (``q" for quadrilateral) by replacing the triangle inequalities by ``quadrilateral" inequalities (Condition \ref{eqn:quad-ineq} in the definition below).
\begin{definition} \label{def:qmetric}
A function $d : \V \times \C \to \Rplus$ is a q-metric iff $\forall v,\vpr \in \V$, and $\forall c, \cpr \in \C$, we have the following:
\begin{enumerate}
\item $d(v,c) \geq 0$
\item $d(v,c) \leq d(v,\cpr) + d(\vpr,\cpr) + d(\vpr,c)$ \label{eqn:quad-ineq}
\end{enumerate}
\end{definition}

The following equivalence result could be of independent interest in problems involving metrics. We make heavy use of it in later sections to prove our results.
\begin{lemma}\label{thm:tri-eq-quad}
If $d$ is a q-metric, then there exists a metric $\dpr$ such that $d(v,c) = \dpr(v,c)$ for all $v \in \V$ and $c \in \C$.
\end{lemma}
\begin{proof}
For all $v,\vpr \in \V$ and $c,\cpr \in \C$,  we define
\begin{align}
\dpr(v,c) = \dpr(c,v) = d(v,c), \\
\dpr(c,\cpr) = \max_{v \in \V} |d(v,c) - d(v,\cpr)|, \\
\dpr(v,\vpr) = \max_{c \in \C} |d(v,c) - d(\vpr,c)|.
\end{align}
Clearly, by the above definitions, and that of a q-metric, for all $x,y \in \C \cup \V$, we have $d(x,y) \ge 0$, $d(x,x) = 0$ and $d(x,y) = d(y,x)$.

Consider $c_1,c_2,c_3 \in \C$. Without loss of generality with respect to $c_1,c_2,c_3$, there exists $u \in \V$ such that 
\begin{align*}
\dpr(c_1,c_3) &= d(u,c_1) - d(u,c_3)& \\
&= d(u,c_1) - d(u,c_2) + d(u,c_2) - d(u,c_3) &\\
&\leq |d(u,c_1) - d(u,c_2)| + |d(u,c_2) - d(u,c_3)| &\\
&\leq \dpr(c_1,c_2) + \dpr(c_2,c_3).&
\end{align*}

Consider $v_1,v_2 \in \V$ and $c \in \C$. Again without loss of generality with respect to $v_1,v_2$, there exists $\cpr \in \C$ such that 
\begin{align*}
\dpr(v_1,v_2) &= d(v_1,\cpr) - d(v_2,\cpr) &\\
&\leq d(v_1,c) + d(v_2,c) & \\
&= \dpr(v_1,c) + \dpr(v_2,c).&
\end{align*}
The inequality in the second line of the above follows by Condition \ref{eqn:quad-ineq} in Definition \ref{def:qmetric}.
Inequalities corresponding to Condition \ref{eqn:tri-ineq} in Definition \ref{def:metric} for triangles given by $v_1,v_2,v_3$ and $v,c_1,c_2$ for all $v,v_1,v_2,v_3 \in \V$ and $c_1,c_2 \in \C$ follow analogously.
\end{proof}

Henceforth, we will deal mainly with \emph{q-metrics} and use the terms \emph{metric} and \emph{q-metric} interchangeably.

\subsection{Distortion}
We say that a metric $d$ is \emph{consistent} with an instance $\sigma$, if whenever any agent $v$ prefers $c$ over $\cpr$, then the her cost for $c$ must be at most her cost for $\cpr$, i.e., $c \succ_v \cpr \implies d(v,c) \leq d(v,\cpr)$. We denote by $\rho(\sigma)$ the set of all metrics $d$ that are consistent with $\sigma$.

The social cost of an alternative is taken as the sum of agent costs for it. For any metric $d$, we define $\phi(c,d) = \sum_{v \in \V} d(v,c)$. For any instance $\sigma$, a consistent metric $d$, and any deterministic social choice rule $f$, define $\Phi(f(\sigma),d) = \phi(f(\sigma),d)$. If $f$ is a randomized social choice rule, we define $\Phi(f(\sigma),d) = \E[\phi(f(\sigma),d)]$.

As mentioned before, we want to measure how close a social choice rule gets to the optimal alternative in terms of social cost. We view every social choice as trying to approximate the optimal alternative, with knowledge of only the agent preference instance $\sigma$, but not the underlying metric cost $d$ that induces $\sigma$. To measure this performance, we take the ratio of the social cost of the alternative chosen by the rule for $\sigma$, and the optimal alternative according to $d$. Distortion \citep{procaccia2006distortion} is then defined as the worst-case value of this quantity over all metrics $d$ that are consistent with $\sigma$:
\begin{align*}
\dist(f,\sigma) = \sup_{d \in \rho(\sigma)} \frac{\Phi(f(\sigma),d)}{\min_{c \in \C} \phi(c,d)}
\end{align*}
In other words, the distortion of a rule $f$ on an instance $\sigma$ is the worst-case ratio of the social cost $\Phi$ of $f(\sigma)$, and that of the optimal alternative. By worst-case we mean the largest value of the above over all possible metrics $d$ that could induce $\sigma$, since $f$ does not know what the true underlying metric is. 
In fact, we look to bound the quantity $\dist(f,\sigma)$ over all possible instances, so as to have a measure of performance for the given rule $f$ independent of the what the instance is, i.e.,  the \emph{worst-case distortion} of $f$.

\subsection{Fairness} \label{subsec:prelims-fairness}
Given an underlying metric, based on the alternative chosen, the costs incurred might vary widely among the agents. We want to formally quantify how ``fair" choosing a particular alternative is. For this purpose, we look at social cost defined as the sum of $k$ largest agent costs, for all $1 \leq k \leq N$. For any metric $d$ and $c \in \C$, we define $\forall 1\leq k \leq N$,
\begin{align*}
\phi_k(c,d) = \max_{S \subseteq \V : |S| = k} \sum_{v \in \V} d(v,c).
\end{align*}
For a deterministic social choice rule $f$, we define $\Phi_k(f(\sigma),d) = \phi_k(f(\sigma),d)$, for all instances $\sigma$ and consistent metrics $d$. If $f$ is a randomized social choice rule, we define $\Phi_k(f(\sigma),d) = \E[\phi_k(f(\sigma),d)]$\footnote{We could define variations of this objective, leading to interesting open questions (see Section \ref{subsec:metriccandidate})}, for all instances $\sigma$ and consistent metrics $d$. We define the fairness-ratio of $f$ as follows:
\begin{align*}
\fratio(f,\sigma) = \sup_{d \in \rho(\sigma)} \max_{1 \leq k \leq N} \frac{\Phi_k(f(\sigma),d)}{\min_{c \in \C} \Phi_k(c,d)}
\end{align*}
The fairness ratio of a rule $f$ on an instance $\sigma$ is a worst-case bound on how well it simultaneously (for all $k$) approximates the social cost given by $\Phi_k$ of $f(\sigma)$, compared to the optimal alternative, over all possible metrics $d$ that could induce $\sigma$, without knowing what the true underlying metric is.

\paragraph{Bounds for general convex costs via the fairness ratio}
Another reason for studying the fairness ratio is that for deterministic social choice rules, a bound on the fairness-ratio translates to an approximation result with respect to any canonical cost function -- a symmetric convex function $F$ of the vector of agent costs such that $F\left(\vec{0}\right) = 0$ and $F$ is non-decreasing in each argument \citep{goel2006simultaneous}. 

For any $c \in \C$, define $\vec{d}(c) = [d(v_1,c), \ldots, d(v_N,c)]$, where $\V = \{v_1,\ldots,v_N\}$.

\begin{theorem}\label{thm:maj-fair}
For any deterministic social choice rule $f$, instance $\sigma$, consistent metric $d$, and canonical cost function $F$, if $~\fratio(f,\sigma) \leq \alpha$, then for any $c \in \C$, 
\begin{align*}
F \left( \frac{\vec{d}(f(\sigma))}{\alpha} \right) \leq F \left( \vec{d}(c) \right).
\end{align*}
\end{theorem}
\begin{proof}
For any vector $\vec{x} \in \mathbbm{R}^N$, let $x_{(1)} \geq x_{(2)} \geq \ldots x_{(i)} \geq \ldots x_{(N)}$ denote its components arranged in some non-decreasing order. For any $1 \leq k \leq N$, define $S_k(\vec{x}) = \sum_{i = 1}^k x_{(i)}$.

A vector $\vec{x}$ is said to be $\alpha$-submajorized by $\vec{y}$ iff $S_k(\vec{x}) \leq \alpha S_k(\vec{y})$, for all $1 \leq k \leq N$.

If $\fratio(f,\sigma) \leq \alpha$, then for any $c \in \C$, we have that $\vec{d}(f(\sigma))$ is $\alpha$-submajorized by $\vec{d}(c)$. This implies that $F \left( \frac{\vec{d}(f(\sigma))}{\alpha} \right) \leq F \left( \vec{d}(c) \right)$ (Theorem 2.3 in \citep{goel2006simultaneous}).
\end{proof}

If a deterministic social choice rule has a fairness ratio of at most $\alpha$, then for all $ p\ge 1$, the $l_p$ norm of the cost vector for the agents under this social choice rule is at most $\alpha$ times the optimum, giving an ``all-norms'' approximation (Corollary \ref{cor:maj-fair-norm}). As special cases, this gives an $\alpha$-approximation for many objective functions such as the sum, the maximum, and the $\ell_2$ norm of the agents costs for an alternative, using $p=$ $1$, $\infty$, and $2$ respectively.

\begin{corollary}\label{cor:maj-fair-norm}
For any deterministic social choice rule $f$, instance $\sigma$, and $p \geq 1$, if $~\fratio(f,\sigma) \leq \alpha$, then for any $c \in \C$, 
\begin{align*}
|| \vec{d}(f(\sigma)) ||_p \leq \alpha || \vec{d}(c) ||_p.
\end{align*}
\end{corollary}

\section{Lower bounds on distortion}\label{sec:lower-bounds}

In this section, we will establish that Ranked Pairs fails to achieve a distortion of at most $3$, contrary to what has been conjectured \citep{anshelevich2015approximating}, thereby falling short of the lower bound on the worst-case distortion of any deterministic rule. We also show a similar result on how tournament/weighted-tournament rules fail to come close to the lower bound of $2$ on the worst-case distortion of any randomized rule.

\begin{figure}[h!]
  \centering
  \def\svgwidth{\columnwidth}
  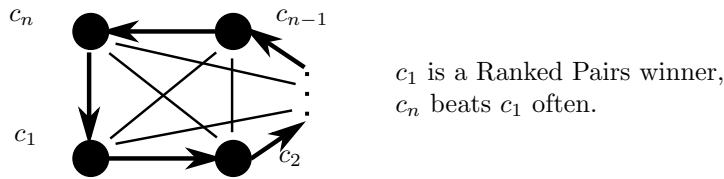
 \vspace{-15pt}
  \caption{Weighted-tournament Graph: A cycle of heavy edges making $c_1$ the Ranked Pairs winner.}
  \label{fig:cycle}
\end{figure}

\subsection{Ranked Pairs}

 In our first result, we will show that the worst-case distortion of Ranked Pairs is at least 5. \citep{anshelevich2015approximating} conjectured that the worst case bound here is 3. This conjecture was based on the result that if the tournament graph does not have cycles of length greater than 4, then the distortion of Ranked Pairs is, in fact, bounded above by 3. 

\begin{theorem}[\citep{anshelevich2015approximating}]\label{thm:anshel}  The distortion of ranked pairs is at most 3, as long as the tournament graph has circumference at most 4.
\end{theorem} 
 
 Assume for a moment that among the set of alternatives $\C$, $c$ is the Ranked Pairs winner, and $\cpr$ is the optimal alternative that minimizes the sum of agent costs. To achieve a large distortion, $\cpr$ must beat $c$ often.  And since $c$ is the Ranked Pairs winner, at the step when $\cpr \to c$ is considered in the Ranked Pairs iteration over edges, a path from $c$ to $\cpr$ must already be in place.
 
One way of achieving this structure is to have $n$ agents, each with a preference ordering that is a different \emph{cyclic permutation} of $c_1,c_2, \ldots, c_n$. $c_1$ is then a Ranked Pairs winner (assuming ties are broken in its favor), and the cycle $c_1 \to c_2, c_2 \to c_3, \ldots, c_{n-1} \to c_n, c_n \to c_1$ has edges of (equal) weight larger than those of edges not on the cycle (Figure \ref{fig:cycle}). The worst case distortion in this case, however, is only 3. 

\subsubsection{Coupling of two sets of cyclic permutations} To achieve a larger distortion, we engineer an overall cyclic structure similar to Figure \ref{fig:cycle} on $2n+1$ alternatives with $n+2$ agents as follows: construct $n$ agents with distinct preference orderings by taking each cyclic permutation of $c_1,\ldots,c_n$ and coupling it with a corresponding permutation of $c_{n+2},\ldots,c_{2n+1}$, pivoted about $c_{n+1}$. We add two agents with the preference exactly $c_1,c_2, \ldots, c_{2n+1}$. 

To understand this coupling, let us look at a related example when $n=2$.

\begin{figure}[h!]
  \centering
  \includegraphics[width = 0.5\columnwidth]{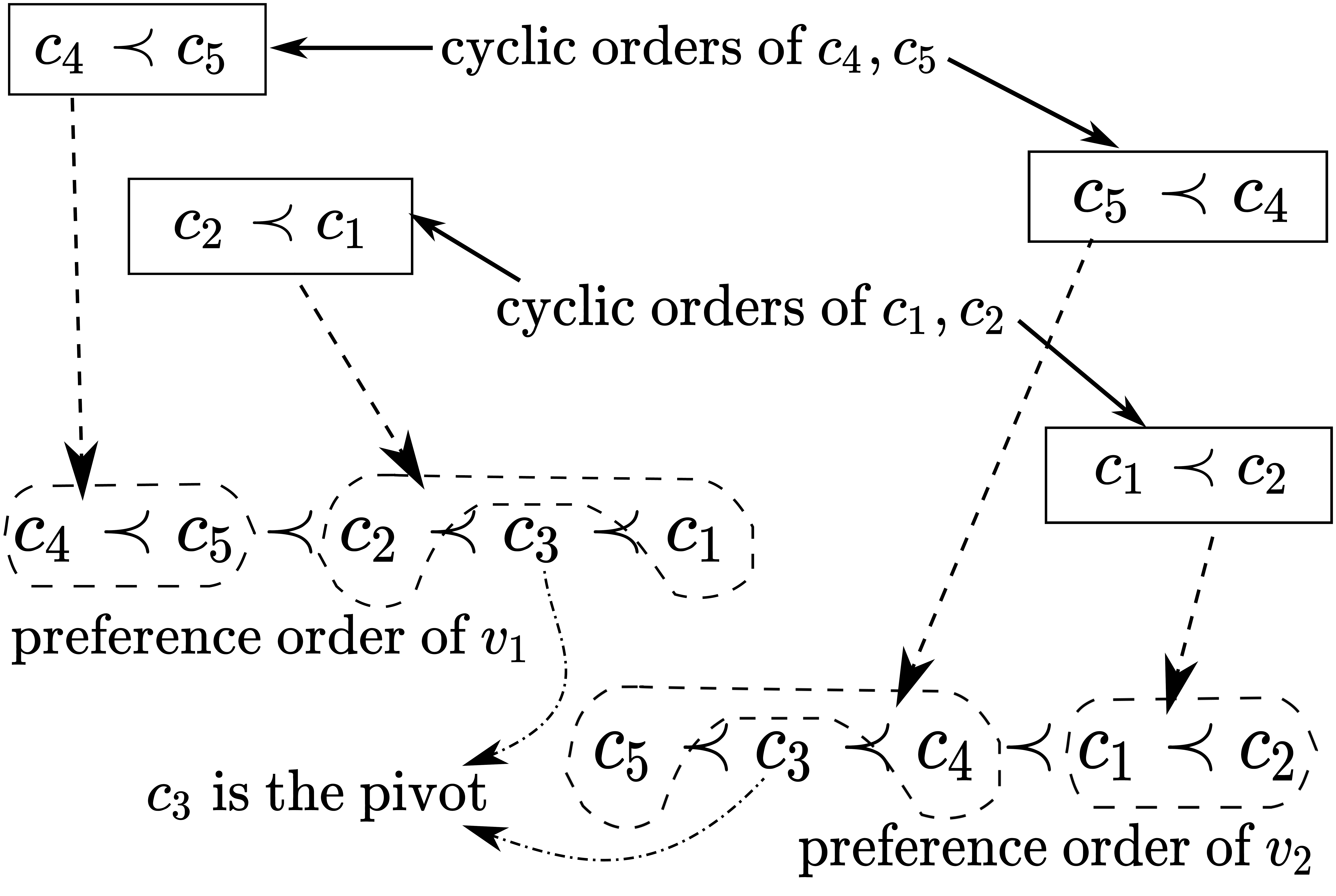}
 \caption{Coupling example: $n=2$}
  \label{fig:coupled}
\end{figure}

\begin{example}
We couple $c_4 \succ c_5$ with $c_2 \succ c_1$, and $c_5 \succ c_4$ with $c_1 \succ c_2$ using $c_3$ as a pivot to get two agents $v_1$ and $v_2$ as in Figure \ref{fig:coupled}.
To make $c_1$ the unique Ranked Pairs winner, we make $n$ copies each of $v_1$ and $v_2$, and add $n+1$ copies of a third agent $v_0$ with preference $c_1 \succ c_2 \succ c_3 \succ c_4 \succ c_5$.
We will see in the proof of Theorem \ref{thm:rpbound} that the following is a valid metric:

\begin{align*}
&\underbrace{c_4 \succ_{v_1} c_5}_\text{$d(v_1,.) = 1$} \succ_{v_1} \underbrace{c_2 \succ_{v_1} c_3}_\text{$d(v_1,.) = 3$} \succ_{v_1} \underbrace{c_1}_\text{$d(v_1,.)=5$} \\
&\underbrace{c_5}_\text{$d(v_2,.) = 0$} \succ_{v_2} \underbrace{c_3 \succ_{v_2} c_4}_\text{$d(v_2,.) = 2$} \succ_{v_2} \underbrace{c_1 \succ_{v_2} c_2}_\text{$d(v_2,.)=4$} \\
&\underbrace{c_1 \succ_{v_0} c_2 \succ_{v_0} c_3 \succ_{v_0} c_4 \succ_{v_0} c_5}_\text{$d(v_0,.) = 2$}
\end{align*}
The ratio of the total costs of $c_1$ and $c_5$ here is $\frac{5(n)+4(n)+2(n+1)}{1(n)+0(n)+2(n+1)} = \frac{11n+2}{3n+2}$, which is more than 3 for $n \geq 3$. This  serves as simple counter-example to the conjecture that Ranked Pairs achieves a distortion of 3. 
\end{example}

This example can be modified to give a sequence of instances that lead to a distortion of $5$ in the limit. In every instance in this sequence, we will see that the Ranked Pairs winner does not depend on how ties are broken.
\begin{theorem} \label{thm:rpbound}
There exists a sequence of instances $\{ \sigma^{(n)}\}_{n \geq 2}$ such that \[\lim_{n \to \infty} \dist(\mbox{Ranked-Pairs},\sigma^{(n)}) = 5.\]
\end{theorem}
\begin{proof}
For each $n \ge 2$, construct an instance $\sigma^{(n)}$ and a corresponding metric $d$ as follows: 
There are $n+2$ agents given by $\V = \{v_0, \vpr_0, v_1, v_2, \ldots, v_n\}$, and $2n+1$ alternatives given by $\C = \{c_1,c_2,\ldots,c_{2n+1}\}$.

Both $v_0$ and $\vpr_0$ have the preference order $c_1 \succ c_2 \succ \ldots \succ c_{2n+1}$, and $d(v_0,c) = d(\vpr_0,c) = 2$ for all $c \in \C$.

For $1 \leq i \leq n$, $v_i$ has the preference order 
\[ \underbrace{c_{n+i+1} \succ ... \succ c_{2n+1}}_\text{$d(v_i,.) = 1$} \succ \underbrace{c_{i+1} \succ ... \succ c_{n+i}}_\text{$d(v_i,.) = 3$} \succ \underbrace{c_1 \succ ... \succ c_i}_\text{$d(v_i,.)=5$} \]. 

Also, define $d$ as follows, for all $c \in \C$:

\begin{align*}
    d(v_i , c_j) = 
\begin{cases}
   1,& \text{if } n+i+1 \leq j \leq 2n+1,\\
   3,& \text{if } i+1 \leq j \leq n+i, \\
   5,& \text{if } 1 \leq j \leq i.
\end{cases}
\end{align*}

First we show that $d$ thus constructed is a valid q-metric. For all $(v,c) \in \V \times \C$, $d(v,c) \geq 0$ is trivially satisfied.
Let $A = \{c_1,c_2,\ldots,c_n\}$ and $B = \{c_{n+1},\ldots,c_{2n+1}\}$. For all $a,\apr \in A$, and $b,\bpr \in B$, and $v \in \V$,

\begin{align*}
|d(v,a) - d(v,b)| \leq 4, \\
|d(v,a) - d(v,\apr)| \leq 2, \\
|d(v,b) - d(v,\bpr)| \leq 2.
\end{align*}

The first holds with equality when $d(v,a) = 5, d(v,b) = 1$, the second when one of $d(v,a), d(v,\apr) $ is 5 and the other is 3, and the third when one of $d(v,b), d(v,\bpr) $ is 3 and the other is 1. We also have

\begin{align*}
d(v,a) + d(v,b) \geq 4, \\
d(v,a) + d(v,\apr) \geq 4, \\
d(v,b) + d(v,\bpr) \geq 2.
\end{align*}

The first holds with equality when $d(v,a) = 3, d(v,b) = 1$ or $d(v,a) = d(v,b) = 2$, the second when $d(v,a) = d(v,\apr)=2 $, and the third when $d(v,b) = d(v,\bpr) = 1$. 
Putting the above inequalities together, we see that $d$ is a valid q-metric since it satisifies Definition \ref{def:qmetric}.

Also from the above, we have $\sum_{v \in \V}d(v,c_1) = 5n+4$ and $\sum_{v \in \V} d(v,c_{2n+1}) = n+4$, and so 
\[ \lim_{n \to \infty} \frac{\sum_{v \in \V}d(v,c_1)}{\sum_{v \in \V}d(v,c_{2n+1})} = 5.\]

We will now show that $c_1$ is the Ranked Pairs winner, irrespective of how ties are broken, in every $\sigma^{(n)}$.

Recall that $w(i,j) = |\{v \in \V : c_i \succ_v c_j \}|$, the strength of edge $c_i \to c_j$ in the weighted-tournament graph obtained from $\sigma_n$. We will first show that for all $1 \leq i \leq 2n$, $w(i,i+1) = n+1$: If $1 \leq i \leq n$, then $w(i,i+1) = n+1$, since $c_i \succ_v c_{i+1}$ for all $v \in \V$ except $v_i$. If $n+1 \leq i \leq 2n$, then $w(i,i+1) = n+1$, since $c_i \succ_v c_{i+1}$ for all $v \in \V$ except $v_{i-n}$.

All other edges $i \to j$ fall into the following cases: \begin{itemize}
\item $i < j-1$: If $j \leq n+1$, then $c_j \succ_{v_k} c_i$ for all $i \leq k \leq j-1$. A similar argument holds when $n+1 \leq i$. If $i \leq n$ and $j \geq n+2$, then $c_j \succ_{v_k} c_i$ at least for $k = i,j-n-1$; 
\item $i > j$, then $c_j \succ_v c_i$ at least for $v \in \{v_0,\vpr_0\}$. 
\end{itemize} 

In all these cases, $c_j \succ_v c_i$ at least for two agents, and thereby $w(i,j) \leq n$. Therefore, the edges $c_i \to c_{i+1}$ for $1 \leq i \leq 2n$ have the largest weights (with no ties) and consequently $c_1$ is the Ranked Pairs winner.
\end{proof}

The \emph{Schulze} method \citep{schulze2003new} also gives priority to edges of larger weight, albeit in a more complicated way. The above result holds for the Schulze rule since it also picks $c_1$ as the winner in the instances constructed.

\begin{corollary}
\[\lim_{n \to \infty} \dist(\mbox{Schulze},\sigma^{(n)}) = 5.\]
\end{corollary}
\begin{proof}
The proof follows from the fact that the Schulze rule also chooses $c_1$ as the winner in every instance in the sequence $\{ \sigma^{(n)}\}_{n \geq 2}$, irrespective of how ties are broken.
\end{proof}

Since methods like Ranked Pairs and the Schulze rule \citep{schulze2003new} fall in the category of weighted-tournament rules (C2 functions), we believe that no weighted-tournament rule can achieve a worst-case distortion of less than 5. 

\begin{conjecture}\label{conj:det-tourney}
Any weighted-tournament rule has a worst-case distortion of at least 5.
\end{conjecture}

 Copeland falls in the category of tournament rules (C1 functions), and we know that Copeland, and other similar rules related to the uncovered set \citep{moulin1986choosing}, achieve a worst-case distortion of 5 \citep{anshelevich2015approximating}. In fact, a lower bound of 5 for the worst-case distortion of Copeland is established via an instance where the tournament graph is a 3 node cycle \citep{anshelevich2015approximating}. It therefore follows that the worst-case distortion of any deterministic tournament rule is at least 5 (since such a rule has no way of distinguishing between the 3 nodes). 

\subsection{Randomized tournament/weighted-tournament rules}
We will now turn our attention to randomized social choice rules. The worst-case distortion in this case is at least 2 \citep{anshelevich2016randomized}. Continuing our discussion on tournament and weighted-tournament rules, we show that in the worst-case randomized tournament/weighted-tournament rules do not get close to the above lower bound. We will construct a sequence of instances where any randomized tournament or weighted-tournament rule achieves a distortion of $3$ in the limit.

\begin{theorem} \label{thm:rand-tourney}
Any randomized tournament or weighted-tournament rule has a worst-case distortion of at least 3.
\end{theorem}
\begin{proof}
Construct an instance $\sigma^{(m)}$ and a corresponding metric as follows:
There are $m+1$ alternatives given by $\C = \{c^*, c_1, c_2, \ldots, c_m \}$. And there are $2m$ agents given by $\V$, and $\V$ is divided into two groups $V = \{v_1,v_2, \ldots, v_m \}$ and $U = \{u_1,\ldots,u_m\}$. 

$v_1$ has the preference order $c^* \succ_{v_1} c_1 \succ_{v_1} \ldots \succ_{v_1} c_m$.
For $2 \leq i \leq m$, agent $v_i$ has the preference order 
$c^* \succ_{v_i} c_i \succ_{v_i} c_{i+1} \succ_{v_i} \ldots \succ_{v_i} c_m \succ_{v_i} c_1 \succ_{v_i} \ldots c_{i-1}$.

$u_1$ has the preference order $c_m \succ_{u_1} \ldots \succ_{u_1} c_1\succ_{u_1} c^*$.
For $2 \leq i \leq m$, agent $v_i$ has the preference order 
 $c_{i-1} \succ_{u_i} \ldots  c_1 \succ_{u_i} c_m \succ_{u_i} \ldots \succ_{u_i} c_{i+1} \succ_{u_i} c_{i} \succ_{u_i} c^*$.
 
Define a metric $d$ as: $\forall v \in V$, $d(v,c^*) = 0$ and $d(v,c) = 2$, $ \forall c \neq c^*$; and $\forall u \in U$, $d(u,c) = 1$, for all $c \in \C$.  We omit the details, but this is indeed a valid metric (see Figure \ref{fig:proof2-pic} for a graphical illustration).
 
For any given distribution $\vec{x}$ over the alternatives $\C$, we must have some alternative $a$ such that $x_a \le \frac{1}{m}$.
In this instance, we have $w(a,b) = |\{v \in V \cup U :a \succ_v b \}| = m$ for all $a \neq b \in \C$, i.e., $\frac{w(a,b)}{|\V|} = 0.5$. Since the tournament/weighted-tournament graph is completely symmetric (since the edge weight is equal to $0.5$ on all directed edges), we can assume without loss of generality that $c^*= a$ \footnote{Tournament/weighted-tournament rules are inherently anonymous, but the winner in this case will depend on hows ties are broken. We can get around this issue by tailoring the constructed instance appropriately, i.e., swapping the roles of $c^*$ and the chosen winner.}.

The expected cost for $c^*$ is

\begin{align*}
\sum_{v \in V} d(v,c^*) + \sum_{u \in U} d(u,c^*) = m(0) + m(1) = m,
\end{align*}

The expected cost for the distribution $\vec{x}$ is

\begin{align*}
\sum_{c \in \C} x_c \sum_{v \in V} d(v,c) + \sum_{c \in \C} x_c \sum_{u \in U} d(u,c)
= \;  x_{c^*}\sum_{t \in V \cup U}d(t,c^*) + \sum_{c \neq c^*} x_c \sum_{t \in V \cup U} d(t,c)
\end{align*}
Since $d(v,c^*) = 0$ for all $v \in V$, and $d(u,c^*) = 1$ for all $u \in U$, we get $\sum_{t \in V \cup U}d(t,c^*) = m$. For any $c \neq c^*$, $t \in V \cup U$, we defined $d(t,c) = 3$, which implies that $\sum_{t \in V \cup U} d(t,c) = 3m$.

Using the above, the expected cost becomes
\begin{align*}
x_{c^*}(m) + \sum_{c \neq c^*} x_c (3m) &=  x_{c^*}(m) + (1 - x_{c^*})(3m)
=  3m - x_{c^*}(2m)\\
&\ge 3m - \frac{1}{m} (2m) = 3m-2.
\end{align*}

Therefore, the distortion ratio is at least
$(3m - 2)/{m} = 3 - {2}/{m}$
which tends to $3$ as $m \to \infty$ . 
\end{proof} 

\begin{figure}
  \centering
  \def\svgwidth{200pt}
  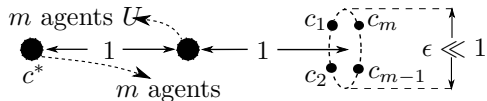
 \vspace{-5pt}
  \caption{Underlying metric in proof of Theorem \ref{thm:rand-tourney}}
  \label{fig:proof2-pic}
\end{figure}

Putting Theorems \ref{thm:rpbound} and \ref{thm:rand-tourney} together, we see that Copeland does at least as well as Ranked Pairs, and randomized tournament/weighted-tournament rules perform no better than Randomized Dictatorship with respect to the worst-case distortion of the sum of costs objective. In the next section, we will show that the upper bounds on the distortion of the sum of costs, for both Copeland and Randomized Dictatorship, hold much more generally.

\subsection{Instance Optimal Distortion}
For any given instance, the single alternative that achieves the least worst-case distortion over all consistent metrics can be found in polynomial time (by solving a polynomial number of linear programs). This follows in a straighforward fashion from the fact that the metric inequalities are linear. The same is true also in the randomized case, perhaps not so directly, in that we can find the optimal distribution over alternatives in polynomial time. The technical details of this claim are provided in Section \ref{subsec:inst-opt-dist} in the appendix. Given the computability of these instance optimal functions, we believe that:

\begin{conjecture}\label{conj:detrandbound}
There exists a deterministic social choice rule that achieves a worst-case distortion of at most 3, and a randomized rule that achieves a worst-case distortion of at most 2.
\end{conjecture}

\section{Fairness in distortion}\label{sec:fairness}
As mentioned before, we introduce a way of quantifying the fairness of social choice rules by using the concept of approximate majorization within the metric distortion framework. One could hope to adapt other notions of fairness like envy-freeness \citep{lipton2004approximately} and leximin \citep{rawls2009theory, barbara1988maximin} (actually leximax since we are dealing with costs) into the distortion framework. While the standard definition of envy-freeness applies to problems involving the division of goods, and requires no inter-personal comparison of utilities, we could perhaps re-purpose it to our setting by looking at the difference between the largest and smallest costs for any given alternative. This quantity in itself cannot be bounded because it is not scale-invariant -- for whenever it is positive, the metric can be scaled to make the envy unbounded. Unfortunately, using this ``envy" as the objective in the distortion ratio is also fruitless (see Example \ref{eg:fairness} below).

We know that Copeland achieves a worst-case distortion of $5$ for objectives given by $\alpha$-percentiles for $\alpha \in [0.5,1)$ \citep{anshelevich2015approximating}. Here the $\alpha$-percentile cost with respect to any alternative $c \in \C$ corresponds to the smallest value $x \in \{d(v,c) | \; v \in \V\}$ for which $\frac{\# \{v \in \V | \; d(v,c) \leq x\}}{N} > \alpha$. It is easy to see that such an upper bound for all $\alpha \in [0,1)$ would subsume our results. However, for $\alpha \in [0,0.5)$, any deterministic rule has unbounded distortion in the worst case. To see this, consider the following example.

\begin{example}\label{eg:fairness}
Consider two alternatives $c_1,c_2$. There are two sets of agents $U$ and $V$ of size $N/2$ each which have preference $c_1 \succ c_2$ and $c_2 \succ c_1$ respectively. Assume without loss of generality that $c_1$ is picked as the winner by the given social choice rule.

Let $d(v,c_1) = 1 \leq d(v,c_2) = 1$ for all $v \in V$, and $d(u,c_2) = \epsilon < d(u,c_1) = 1$ for all $u \in U$. By invoking Lemma \ref{thm:tri-eq-quad}, we can see that this gives us a consistent metric. 

For $\alpha \in [0,0.5)$, we have that the $\alpha$-percentile costs for $c_1$ and $c_2$ are $1$ and $\epsilon$ respectively. The distortion ratio is then $\frac{1}{\epsilon}$, which goes to $\infty$ as $\epsilon \to 0$.

Now assume $c_2$ is picked as the winner. The maximum envy in this case is $1 - \epsilon$. And in the case of $c_1$ the maximum envy is $0$, leading to an unbounded ratio.
\end{example}

The above example also shows us why a leximax comparison does not work.  Ordering the costs for $c_1$ and $c_2$ in non-decreasing order, let us compare the costs at the first position at which these orders differ. At the $(N/2 +1)$-th position, the cost for $c_1$ is $1$, and that for $c_2$ is $\epsilon$, leading to an unbounded ratio as $\epsilon \to 0$.

\subsection{Bounding the fairness ratio of Copeland rule}
In this section, we show that Copeland achieves a fairness ratio of at most $5$. Besides Copeland, other weighted-tournament rules such as those selecting winners from the minimal covering set, the bipartisan set, banks set, or any other subset of the uncovered set, also achieve a fairness-ratio of at most $5$.

\begin{theorem}\label{thm:copeland-fairness}
For any instance $\sigma$, if $x$ is the Copeland winner, and $z$ is any other alternative, then
\begin{align*}
\fratio(Copeland,\sigma) = \sup_{d \in \rho(\sigma)} \max_{1 \leq k \leq N} \frac{\max_{S \subseteq \V : |S| = k} \sum_{v \in S} d(v,x)}{\max_{S \subseteq \V : |S| = k} \sum_{v \in S} d(v,z)} \leq 5.
\end{align*}
\end{theorem}

\begin{proof}
Fix any $k \in \{1,2,\dots, N\}$. 

For any $a,b \in \C$, denote the set of agents that prefer $a$ over $b$ by $G_{ab} = \{ v \in \V : a \succ_v b \}$. 

And for any $t \in \C$, let \[S_t \triangleq \arg \max_{S \subseteq \V : |S| = k} \sum_{v \in S} d(v,t).\] 

Since $x$ is the Copeland winner, we know, from the connection to the uncovered set \citep{moulin1986choosing}, that either
\begin{inparaenum}[(A)]
\item $|G_{xz}| \geq \frac{N}{2}$, or
\item $\exists y \in \C$, such that $|G_{xy}| \geq \frac{N}{2}$ and $|G_{yz}| \geq \frac{N}{2}$
\end{inparaenum}. We deal with each case separately.

\textit{Case (A):} 
Let $g : \V \to \V$ be any one-one map such that if $v \in S_x \setminus G_{xz}$ then $g(v) \in G_{xz}$. One such map exists because $|S_x \setminus G_{xz}| \leq |\V \setminus G_{xz}| \leq \frac{N}{2} \leq |G_{xz}|$. Let $A = S_x \cap G_{xz}$ and $B = S_x \setminus G_{xz}$.

\begin{align*}
\sum_{v \in S_x} d(v,x) &= \sum_{v \in A} d(v,x) +  \sum_{v \in B} d(v,x) \\
& \leq \sum_{v \in A} d(v,z) +  \sum_{v \in B} d(v,x) \\
&\leq \sum_{v \in A} d(v,z) + \sum_{v \in B}\left( d(g(v),x) + d(g(v),z) + d(v,z) \right) \\
&= \sum_{v \in S_x} d(v,z)  +  \sum_{v \in B} \left( d(g(v),x) + d(g(v),z) \right) \\
&\leq \sum_{v \in S_x} d(v,z) +  \sum_{v \in B} 2d(g(v),z) \\
& \leq 3\sum_{v \in S_z} d(v,z).
\end{align*}
In the above sequence, the first inequality follow from the fact that $v \in A \implies v \in G_{xz}$ and for any $v \in G_{xz}$, $d(v,x) \leq d(v,z)$ by definition. The second inequality follows after invoking Condition \ref{eqn:quad-ineq} from Definition \ref{def:qmetric}. The third inequality is true because $g(v) \in G_{xz}$ by definition, and the fourth because for any $S \subseteq \V$ such that $|S| \leq k$, $\sum_{v \in S} d(v,z) \leq \sum_{v \in S_z} d(v,z)$ by the definition of $S_z$.

\textit{Case (B):} 
Let $h : \V \to \V$ be any one-one map such that if $v \in S_x \setminus G_{xy}$ then $h(v) \in G_{xy}$. One such map exists because $|S_x \setminus G_{xy}| \leq |\V \setminus G_{xy}| \leq \frac{N}{2} \leq |G_{xy}|$. Let $P = S_x \cap G_{xy}$ and $Q = S_x \setminus G_{xy}$

\begin{align*}
\sum_{v \in S_x} d(v,x)  &= \sum_{v \in P} d(v,x) +  \sum_{v \in Q} d(v,x)  \\
& \leq  \sum_{v \in P} d(v,y) +  \sum_{v \in Q} d(v,x)  \\
  &\leq  \sum_{v \in P} (d(h(v),y) + d(h(v),z) + d(v,z) ) \\
  & \quad+  \sum_{v \in Q} \left( d(h(v),x) + d(h(v),z) + d(v,z) \right) \\
  &=  \sum_{v \in P} d(h(v),y) +  \sum_{v \in Q}d(h(v),x)  + \sum_{v \in S_x}d(v,z) + \sum_{v \in S_x}d(h(v),z)\\
  &\leq  \sum_{v \in P} d(h(v),y) +  \sum_{v \in Q}d(h(v),y)  + \sum_{v \in S_x}d(v,z) + \sum_{v \in S_x}d(h(v),z) \\
 & =  \sum_{v \in S_x} d(h(v),y) + \sum_{v \in S_x}d(v,z) + \sum_{v \in S_x}d(h(v),z)  \\
 & \leq  \sum_{v \in S_y} d(v,y) + 2 \sum_{v \in S_z} d(v,z) \leq  5\sum_{v \in S_z} d(v,z)  
\end{align*}
In the above sequence, the first inequality follow from the fact that $v \in P \implies v \in G_{xy}$ and for any $v \in G_{xy}$, $d(v,x) \leq d(v,y)$ by definition. The second inequality follows after invoking Condition \ref{eqn:quad-ineq} from Defintion \ref{def:qmetric}. The third inequality is true because $h(v) \in G_{xy}$ by definition, and the fourth because for any $t \in \C$ and $S \subseteq \V$ such that $|S| \leq k$, $\sum_{v \in S} d(v,t) \leq \sum_{v \in S_t} d(v,t)$ by the definition of $S_t$. The last follows from the fact that $\sum_{v \in S_y} d(v,y) \leq 3\sum_{v \in S_z} d(v,z)$ by case (a) above.
\end{proof}

The fact that the inequality in Theorem \ref{thm:copeland-fairness} above is tight follows from the known example \citep{anshelevich2015approximating} in which Copeland achieves a distortion of $5$ with respect to the sum of costs objective.

As mentioned before Copeland also does well with respect to other objectives such as median and $\alpha$-percentiles for $\alpha \in [0.5,1)$ \citep{anshelevich2016randomized}. These functions are not convex and hence do not fall under the category of functions that can be approximated with the help of the fairness ratio. An interesting question is to characterize the entire class of functions for which Copeland achieves a constant factor bound on the distortion.

\subsection{Randomized Dictatorship}
For randomized rules, the connection of the fairness-ratio to convex cost functions does not hold in terms of the expectation variants of the quantities involved. However, the fairness ratio in its own right is a generalization of both max-min fairness and total cost minimization, and is hence worth studying in the case of randomized social choice rules.

Our last result is that Randomized Dictatorship, which achieves a worst-case distortion of $3$, also achieves a fairness ratio of $3$ in expectation.
\begin{theorem}\label{thm:rand-dict-fairness}
For any instance $\sigma$, alternative $y$, and $X$ chosen according to Randomized Dictatorship
\begin{align*}
\sup_{d \in \rho(\sigma)} \max_{1 \leq k \leq N} \frac{\E [\max_{S \subseteq \V : |S| = k} \sum_{v \in S} d(v,X)]}{\max_{S \subseteq \V : |S| = k} \sum_{v \in S} d(v,y)} \leq 3.
\end{align*}
\end{theorem}
\begin{proof}
Fix any $k \in \{1,2,\ldots,N\}$, and any alternative $y \in \C$.

For all $c \in \C$, denote the set of agents with $c$ as their top choice by $V_c = \{v \in \V : c \succ_v \cpr, \;\forall \cpr \neq c\} $, and the size of this set as $N_c = |V_c|$. Let the total number of agents be given by $N = |\V|$. 

For any alternative $c \in \C$, denote the set of agents that maximize the sum of $k$ costs for it by $S_c = \arg \max_{S \subseteq \V : |S| = k} \sum_{v \in S} d(v,c)$.

For all $c \in \C$, by the triangle inequality we have
\begin{align}\label{eqn:top}
\nonumber \sum_{v \in S_c} d(v,c) &\leq \sum_{v \in S_c} (d(v,y) + d(c,y)) \\
& \leq k \cdot d(c,y) + \sum_{v \in S_y} d(v,y),
\end{align}
where the second inequality follows by the definition of $S_y$.

For any $c \in \C$, if $v \in V_c$, we have $d(v,c) \leq d(v,y)$ (by the definition of $V_c$), and $d(v,c) + d(v,y) \geq d(c,y)$ (by the triangle inequality), which together imply $d(v,y) \geq \frac{d(c,y)}{2}$. Therefore,

\begin{align}\label{eqn:middle}
\sum_{v \in \V} d(v,y) = \sum_{c \in \C} \sum_{v \in V_c} d(v,y) \geq \sum_{c \in \C} N_c \frac{d(c,y)}{2}.
\end{align}

Consequently, we get

\begin{align}\label{eqn:bottom}
\sum_{v \in S_y} d(v,y) \geq \frac{k}{N}\sum_{v \in \V} d(v,y) \geq \frac{k}{N}\sum_{c \in \C} N_c \frac{d(c,y)}{2},
\end{align}
where the first inequality follows by the definition of $S_y$, and the second from the inequality in \ref{eqn:middle}. We can now bound the expected distortion as follows:
\begin{align*}
&\frac{\E [\max_{S \subseteq \V : |S| = k} \sum_{v \in S} d(v,X)]}{\max_{S \subseteq \V : |S| = k} \sum_{v \in S} d(v,y)} 
=  \frac{\sum_{c \in \C} \frac{N_c}{N} \sum_{v \in S_c} d(v,c)}{\sum_{v \in S_y} d(v,y)} \\
&\leq   \frac{\sum_{c \in \C} \frac{N_c}{N}(k \cdot d(c,y) + \sum_{v \in S_y} d(v,y))}{\sum_{v \in S_y} d(v,y)}  \\
&\leq   1 + \frac{\sum_{c \in \C} \frac{N_c}{N} k \cdot d(c,y)}{\sum_{v \in S_y} d(v,y)} 
\leq   1 + \frac{\frac{k}{N} \sum_{c \in \C} N_c \cdot d(c,y)}{\frac{k}{N}\sum_{c \in \C} N_c \frac{d(c,y)}{2}}  
\leq   3
\end{align*}

We have made use of \ref{eqn:top} and \ref{eqn:bottom} to get the first and third inequalities in the above sequence.
\end{proof}

To make amply clear that the bound on the fairness ratio of randomized does not extend to convex functions, we will now look at an example where the distortion of randomized dictatorship is unbounded when the objective used is the square of the sum of costs.

\begin{example}\label{eg:rand-dict-fairness}
Consider two agents $\{c_1,c_2\}$ and $N_1 + N_2$ agents in total divided into two groups $V_1$ and $V_2$. Assume all of the above are points in $\mathbbm{R}$.
$c_1 = 0$ and $c_2 = 1$. Every agent in $V_1$ is at $0$, and every agent in $V_2$ is at $1$. Also assume that $N_1 \geq N_2$.

Let $C(\vec{x}) = \left( \sum_i x_i \right)^2$. With $C$ as the cost objective, the optimal alternative is $c_1$. The social cost of this alternative is equal to $N_2^2$.

Randomized dictatorship chooses $c_1$ with probability $\frac{N_1}{N_1 + N_2}$ and $c_2$ with probability $\frac{N_2}{N_1 + N_2}$. The expected cost is equal to 
\[N_2^2 \frac{N_1}{N_1 + N_2} + N_1^2 \frac{N_2}{N_1 + N_2} = N_1 N_2. \]

The distortion ratio is $N_1N_2/ N_2^2 = N_1/N_2$ which is unbounded in the limit $N_1 \to \infty$ for any fixed $N_2$.
\end{example}

The fact that the distortion of Randomized Dictatorship with respect to convex cost objectives is unbounded makes the fairness properties of Copeland even more interesting. It seems surprising that a simple rule like Copeland can approximate the optimal alternative over a very general class of cost functions. 

\section{Conclusions}
In this paper, we further the understanding of the performance of social choice rules under metric preferences with respect to the distortion measure. We provide lower bounds on worst-case distortion for deterministic rules such as Ranked Pairs and Schulze, and randomized tournament/weighted-tournament rules. We introduce a framework to study the fairness properties of social choice rules within the distortion framework, and provide low constant-factor upper bounds on the fairness ratios of some well known mechanisms like Copeland and Randomized Dictatorship. In particular, what stands out is that Copeland not only achieves the best known upper bound for deterministic rules, but also simultaneously approximates a large class of cost functions.

\bibliographystyle{plainnat}
\bibliography{metric-distortion_arxiv}

\section{Appendix}
\subsection{Instance-optimal Distortion}\label{subsec:inst-opt-dist}
For any instance $\sigma$, imagine $c$ is the alternative chosen as the outcome of a social choice rule, and $\cpr$ is the candidate with minimum sum cost with respect to the metric that maximizes distortion. It then follows, very directly from the work of \citep{anshelevich2015approximating}, that the value of the distortion can be found using the following Linear program:

\begin{lp}\label{lp:1}
\begin{equation*}
\begin{aligned}
& A(c,c^\prime,\sigma) \triangleq & \max 
& \; \sum_{v \in \V} d(c,v) \\
& & \text{subject to}
& \; \sum_{v \in \V} d(c^\prime,v) = 1 \\
&&& d \in \rho(\sigma)
\end{aligned}
\end{equation*}
\end{lp}
This is a linear program because $d \in \rho(\sigma)$ if and only if $d$ satisfies the inequalities in Definition 2, and is consistent with $\sigma$. Consistency can also be captured by linear inequalities as mentioned earlier.

We can then define an instance-optimal deterministic choice function $\OPTdet$ which chooses the alternative that minimizes the maximum distortion as 
\begin{equation*}
\OPTdet(\sigma) = \arg \min_{c \in \C} \max_{c^\prime \in \C} A(c,c^\prime,\sigma)
\end{equation*}

The map $\OPTdet$ can be computed in polynomial time, since it involves solving a linear program per pair of candidates. We conjecture that this method always achieves a distortion of not more than 3. We believe that this is an interesting combinatorial problem that is worth looking at.
\begin{conjecture}\label{conj:detalg3bound}
$\OPTdet$ achieves a distortion of not more than 3.
\end{conjecture}

We can define a similar randomized rule $\OPTrand$ that finds the instance-optimal alternative. Surprisingly, this can also be computed in polynomial time, and in what follows we show how.

If, for an instance $\sigma$, a distribution $\{x_c\}_{c \in \C}$ ($\sum_{c \in C} x_c = 1$) over $\C$ is chosen as the social outcome, and an alternative $c^\prime$ maximizes distortion, its value can be found by solving the following Linear Program: 

\begin{lp}\label{lp:randalg}
\begin{equation*}
\begin{aligned}
& A(\vec{x},c^\prime,\sigma) \triangleq & \max 
& \; \sum_{c \in \C} x_c \sum_{v \in \V} d(c,v) \\
& & \text{subject to}
& \; \sum_{v \in \V} d(c^\prime,v) = 1 \\
&&& d \in \rho(\sigma).
\end{aligned}
\end{equation*}
\end{lp}

Let $\Delta(\C)$ be the set of all distributions over $\C$. Then we can define $\OPTrand$ to choose the alternative corresponding to the \emph{minimax} solution as follows:

\begin{align}\label{eqn:randminimax-metricresp}
\OPTrand(\sigma) = \arg \min_{\vec{x} \in \Delta(\C)} \max_{c^\prime \in C} A(\vec{x},c^\prime)
\end{align}

We will impose additional constraints on the metrics we consider. We say that a metric $d$ is normal if and only if $\min_{c^\prime \in C} \sum_{v \in V} d(c^\prime,v) = 1$. We will denote by $\theta(\sigma)$ the set of all normal metrics that are consistent with the instance $\sigma$.

This minimax problem above can be cast into a minimization problem in the following way:
\begin{problem}[Minimax problem]\label{prob:minimax}
\begin{equation*}
\begin{aligned}
& &\mbox{minimize } & \; \gamma  &\\
& &\text{subject to} & \; \sum_{c \in \C} x_c = 1& \\
&& &\sum_{c \in \C} x_c \sum_{v \in \V} d(c,v) \leq \gamma & \forall d \in \theta(\sigma)
\end{aligned}
\end{equation*}
\end{problem}

We can use binary search to do a polynomial time reduction of the above to checking feasibility over $\gamma$. We know that $3$ is an upper bound from the fact that Randomized Dictatorship achieves a worst-case distortion of $3$.
\begin{problem}[Feasibility]\label{prob:feas}
Given an instance $\sigma$ and a $\gamma \in [1, 3]$, is it feasible for the Minimax Problem (Problem \ref{prob:minimax})? If so, find $\vec{x}$ such that $\sum_{c \in \C} x_c = 1$ and $\sum_{c \in \C} x_c \sum_{v \in \V} d(c,v) \leq \gamma$,  $\forall d \in \theta(\sigma)$.
\end{problem}

Given $\gamma \in [1,3]$, denote by $F_\gamma$ the convex feasible region of $\vec{x}$ determined by the following inequalities:
\begin{enumerate}[(a)]
\item $\vec{x} \in \Delta(\C)$,
\item$\forall d \in \theta(\sigma)$, $\sum_{c \in \C} x_c \sum_{v \in \V} d(c,v) \leq \gamma$.
\end{enumerate} 
Note that the set of inequalities given by (b) is uncountable.

To solve Problem \ref{prob:feas}, we can make use of the following separation oracle.
\begin{problem}[Separation Oracle]\label{prob:seporacle}
Given $\gamma$ and $\vec{x}$, either claim that $x \in F_\gamma$, or find a $d \in \theta(\sigma)$ such that $\sum_{c \in \C} x_c \sum_{v \in \V} d(c,v) > \gamma$.
\end{problem}

\begin{definition}[$c^\prime$-normal]\label{def:cnormal}
A metric $d$ is called $c^\prime$-normal if and only if 
\begin{enumerate}[(i)]
\item $\min_{c^\prime \in \C} \sum_{v \in \V} d(c^\prime,v) = 1$, and,
\item $\min_{c \neq c^\prime}d(c^\prime,v) \geq 1$.
\end{enumerate}
\end{definition}
A normal metric $d$ has to be $c^\prime$-normal for some $c^\prime$. To solve Problem \ref{prob:seporacle}, we make use of the following two observations:
\begin{itemize}
\item For a given $\vec{x}$, it belongs to $F$ if $\forall c^\prime \in \C$, $\sum_{c \in \C} x_c \sum_{v \in \V} d(c,v) \leq \gamma$ for all consistent, $c^\prime$-normal metrics. 
\item And if $\sum_{c \in \C} x_c \sum_{v \in V} d(c,v) > \gamma$ for all consistent, normal metrics $d$, then it must be so for some consistent $c^\prime$-normal metric.
\end{itemize}

As a result, the separation oracle problem can be solved by solving the following Linear Program \ref{lp:seporacle} for all $c^\prime \in \C$, for the $\vec{x}$ given to the oracle.

\begin{lp}\label{lp:seporacle}
\begin{equation*}
\begin{aligned}
& A(\vec{x},c^\prime) \triangleq & \max 
& \; \sum_{c \in C} x_c \sum_{v \in V} d(c,v) \\
& & \text{subject to}
& \; \sum_{v \in V} d(c^\prime,v) = 1 \\
&&& d \mbox{ is a consistent, } c^\prime \mbox{-normal metric.}
\end{aligned}
\end{equation*}
\end{lp}

\begin{theorem}
The Minimax Problem (Problem \ref{prob:minimax}) can be solved by solving a polynomial number of linear programs of type \ref{lp:seporacle}.
\end{theorem}
This leads us to another interesting direction for future research -- analyzing this polynomial time LP-based algorithm could give us a worst-case distortion of $2$. 

\begin{conjecture}\label{conj:randalg2bound}
$\OPTrand$ achieves a distortion of not more than 2.
\end{conjecture}

\subsection{Distortion of Randomized social choice rules: Metric response vs. Candidate response} \label{subsec:metriccandidate}

\subsubsection{Metric Response}
The notion of distortion we defined for randomized functions can be described in terms of the following game:

We have a sequential game, where the social choice function, denoted by the player $ALG$, plays first choosing a distribution $\{x_c\}_{c \in C}$ over all candidates, followed by an adversarial player $ADV$ who chooses a metric $d$, and automatically the best possible choice of $c^\prime$ for that metric. The payoffs are given by a zero sum game where the utility of $ADV$ is exactly the value of the distortion achieved. The corresponding \emph{minimax} solution is then given by Equation \ref{eqn:randminimax-metricresp}. 

We call the above the \emph{metric response} game, and the value of this game as the \emph{distortion under metric response}.

\subsubsection{Candidate Response}
We could think of variants of the above game where the adversary is more powerful. For example, consider the following:

$ALG$ plays first choosing a distribution $\{x_c\}_{c \in C}$ over all candidates. $ADV$ then follows by choosing a single candidate $c^\prime$.

Let's say that $\sigma$ is the instance at hand. For each draw $c$ from the distribution chosen by $ALG$, and given $c^\prime$, nature picks a metric so that the utility of $ADV$ is $A(c,c^\prime,\sigma)$, and that of $ALG$ is $-A(c,c^\prime,\sigma)$. 

The corresponding \emph{minimax} solution is then given by 
\begin{align}\label{eqn:randminimax-candidateresp}
\min_{\vec{x} \in \Delta(\C)} \max_{c^\prime \in C} A(c,c^\prime,\sigma)
\end{align}

We call the above the \emph{candidate response} game, and the value of this game as the \emph{distortion under candidate response}.

Since the adversary under \emph{candidate response} is stronger than that under \emph{metric response}, the lower bounds for distortion under metric response carry over to the candidate response. For example:
\begin{theorem}
No randomized tournament rule can achieve a distortion ratio of less than $3$ for all problem instances.
\end{theorem}
\begin{proof}
Follows from Theorem \ref{thm:rand-tourney}.
\end{proof}

However, these lower bounds need not be tight. Characterizing tighter lower bounds here is an open question.

Another interesting open problem under \emph{candidate response} is characterizing achievable tight upper bounds on distortion. For example, does Randomized Dictatorship always achieve a distortion of at most 3 in this case?

\end{document}

%% file: mydefs.tex
\newtheorem{conjecture}{Conjecture}
\newtheorem{definition}{Definition}

\newtheorem{problem}{Problem}
\newtheorem{corollary}{Corollary}

\newtheorem{theorem}{Theorem}
\newtheorem{result}{Result}

\usepackage{xcolor}

\newcommand{\C}{\mathcal{C}}
\newcommand{\V}{\mathcal{V}}

\newcommand{\s}{\mathcal{S}}

\newcommand{\vpr}{v^\prime}

\newcommand{\Rplus}{\mathcal{R}_{\geq 0}}

\newcommand{\cpr}{c^\prime}
\newcommand{\apr}{a^\prime}
\newcommand{\bpr}{b^\prime}

\newcommand{\E}{\mathbb{E}}
\usepackage{bbm}

\newtheorem{lemma}{Lemma}
\newtheorem{lp}{Linear Program}

\newtheorem{example}{Example}

\newcommand{\dpr}{d^\prime}

\newcommand{\dist}{\mathrm{dist}}
\newcommand{\fratio}{\mathrm{fairness}}
\newcommand{\OPTdet}{\mathrm{OPT_{det}}}
\newcommand{\OPTrand}{\mathrm{OPT_{rand}}}

%% file: graphmetric2.pdf_tex
\begingroup%
  \makeatletter%
  \providecommand\color[2][]{%
    \errmessage{(Inkscape) Color is used for the text in Inkscape, but the package 'color.sty' is not loaded}%
    \renewcommand\color[2][]{}%
  }%
  \providecommand\transparent[1]{%
    \errmessage{(Inkscape) Transparency is used (non-zero) for the text in Inkscape, but the package 'transparent.sty' is not loaded}%
    \renewcommand\transparent[1]{}%
  }%
  \providecommand\rotatebox[2]{#2}%
  \ifx\svgwidth\undefined%
    \setlength{\unitlength}{1761.00995402bp}%
    \ifx\svgscale\undefined%
      \relax%
    \else%
      \setlength{\unitlength}{\unitlength * \real{\svgscale}}%
    \fi%
  \else%
    \setlength{\unitlength}{\svgwidth}%
  \fi%
  \global\let\svgwidth\undefined%
  \global\let\svgscale\undefined%
  \makeatother%
  \begin{picture}(1,0.18432219)%
    \put(0,0){\includegraphics[width=\unitlength,page=1]{graphmetric2.pdf}}%
    \put(0.30053013,0.02674607){\color[rgb]{0,0,0}\makebox(0,0)[lt]{\begin{minipage}{0.10902836\unitlength}\raggedright $c_3,v_3$\end{minipage}}}%
    \put(0.49042118,0.02674607){\color[rgb]{0,0,0}\makebox(0,0)[lt]{\begin{minipage}{0.07268557\unitlength}\raggedright $v_2$\end{minipage}}}%
    \put(0.22693597,0.11942018){\color[rgb]{0,0,0}\makebox(0,0)[lt]{\begin{minipage}{0.07268559\unitlength}\raggedright $v_1$\end{minipage}}}%
    \put(0.41228419,0.11942018){\color[rgb]{0,0,0}\makebox(0,0)[lt]{\begin{minipage}{0.05451418\unitlength}\raggedright $c_2$\end{minipage}}}%
    \put(0.03068493,0.11760304){\color[rgb]{0,0,0}\makebox(0,0)[lt]{\begin{minipage}{0.05451418\unitlength}\raggedright $c_1$\end{minipage}}}%
    \put(-0.51276124,0.38357685){\color[rgb]{0,0,0}\makebox(0,0)[lt]{\begin{minipage}{0.01817139\unitlength}\raggedright \end{minipage}}}%
    \put(0.58918747,0.17794982){\color[rgb]{0,0,0}\makebox(0,0)[lt]{\begin{minipage}{0.46986027\unitlength}\raggedright $v_2$'s cost for $c_1$ $= ~ 3$\\  = the length of the shortest path.  \end{minipage}}}%
    \put(0,0){\includegraphics[width=\unitlength,page=2]{graphmetric2.pdf}}%
  \end{picture}%
\endgroup%

%% file: cycle2.pdf_tex
\begingroup%
  \makeatletter%
  \providecommand\color[2][]{%
    \errmessage{(Inkscape) Color is used for the text in Inkscape, but the package 'color.sty' is not loaded}%
    \renewcommand\color[2][]{}%
  }%
  \providecommand\transparent[1]{%
    \errmessage{(Inkscape) Transparency is used (non-zero) for the text in Inkscape, but the package 'transparent.sty' is not loaded}%
    \renewcommand\transparent[1]{}%
  }%
  \providecommand\rotatebox[2]{#2}%
  \ifx\svgwidth\undefined%
    \setlength{\unitlength}{969.79337437bp}%
    \ifx\svgscale\undefined%
      \relax%
    \else%
      \setlength{\unitlength}{\unitlength * \real{\svgscale}}%
    \fi%
  \else%
    \setlength{\unitlength}{\svgwidth}%
  \fi%
  \global\let\svgwidth\undefined%
  \global\let\svgscale\undefined%
  \makeatother%
  \begin{picture}(1,0.22123007)%
    \put(0,0){\includegraphics[width=\unitlength,page=1]{cycle2.pdf}}%
    \put(0.040436,0.27136957){\color[rgb]{0,0,0}\makebox(0,0)[lt]{\begin{minipage}{0.0565658\unitlength}\raggedright \end{minipage}}}%
    \put(0.02649079,0.21018879){\color[rgb]{0,0,0}\makebox(0,0)[lt]{\begin{minipage}{0.16749189\unitlength}\raggedright $c_n$\end{minipage}}}%
    \put(0.31883366,0.20978441){\color[rgb]{0,0,0}\makebox(0,0)[lt]{\begin{minipage}{0.17499147\unitlength}\raggedright $c_{n-1}$\end{minipage}}}%
    \put(0.03219881,0.07794821){\color[rgb]{0,0,0}\makebox(0,0)[lt]{\begin{minipage}{0.10332832\unitlength}\raggedright $c_1$\end{minipage}}}%
    \put(0.32202386,0.05528538){\color[rgb]{0,0,0}\makebox(0,0)[lt]{\begin{minipage}{0.10082843\unitlength}\raggedright $c_2$\end{minipage}}}%
    \put(0,0){\includegraphics[width=\unitlength,page=2]{cycle2.pdf}}%
    \put(0.45016954,0.15319905){\color[rgb]{0,0,0}\makebox(0,0)[lt]{\begin{minipage}{0.56565806\unitlength}\raggedright $c_1$ is a Ranked Pairs winner, \\ $c_n$ beats $c_1$ often.\end{minipage}}}%
    \put(-0.01941923,0.31121263){\color[rgb]{0,0,0}\makebox(0,0)[lt]{\begin{minipage}{1.09596246\unitlength}\raggedright \end{minipage}}}%
  \end{picture}%
\endgroup%

%% file: proof2-pic.pdf_tex
\begingroup%
  \makeatletter%
  \providecommand\color[2][]{%
    \errmessage{(Inkscape) Color is used for the text in Inkscape, but the package 'color.sty' is not loaded}%
    \renewcommand\color[2][]{}%
  }%
  \providecommand\transparent[1]{%
    \errmessage{(Inkscape) Transparency is used (non-zero) for the text in Inkscape, but the package 'transparent.sty' is not loaded}%
    \renewcommand\transparent[1]{}%
  }%
  \providecommand\rotatebox[2]{#2}%
  \ifx\svgwidth\undefined%
    \setlength{\unitlength}{1515.66881075bp}%
    \ifx\svgscale\undefined%
      \relax%
    \else%
      \setlength{\unitlength}{\unitlength * \real{\svgscale}}%
    \fi%
  \else%
    \setlength{\unitlength}{\svgwidth}%
  \fi%
  \global\let\svgwidth\undefined%
  \global\let\svgscale\undefined%
  \makeatother%
  \begin{picture}(1,0.17011071)%
    \put(0,0){\includegraphics[width=\unitlength,page=1]{proof2-pic.pdf}}%
    \put(0.21872786,0.10201542){\color[rgb]{0,0,0}\makebox(0,0)[lt]{\begin{minipage}{0.08445117\unitlength}\raggedright $1$\end{minipage}}}%
    \put(0,0){\includegraphics[width=\unitlength,page=2]{proof2-pic.pdf}}%
    \put(0.82210131,0.10699201){\color[rgb]{0,0,0}\makebox(0,0)[lt]{\begin{minipage}{0.19755544\unitlength}\raggedright $\epsilon \ll 1$\end{minipage}}}%
    \put(0.06581093,0.05888501){\color[rgb]{0,0,0}\makebox(0,0)[lt]{\begin{minipage}{0.07389477\unitlength}\raggedright $c^*$\end{minipage}}}%
    \put(0.71834699,0.15419418){\color[rgb]{0,0,0}\makebox(0,0)[lt]{\begin{minipage}{0.0723867\unitlength}\raggedright $c_m$\end{minipage}}}%
    \put(0.72051922,0.05578568){\color[rgb]{0,0,0}\makebox(0,0)[lt]{\begin{minipage}{0.12582413\unitlength}\raggedright $c_{m-1}$\end{minipage}}}%
    \put(0.59917287,0.15321396){\color[rgb]{0,0,0}\makebox(0,0)[lt]{\begin{minipage}{0.0740833\unitlength}\raggedright $c_1$\end{minipage}}}%
    \put(0.60022849,0.04780079){\color[rgb]{0,0,0}\makebox(0,0)[lt]{\begin{minipage}{0.08599687\unitlength}\raggedright $c_2$\end{minipage}}}%
    \put(0.04017397,0.16776669){\color[rgb]{0,0,0}\makebox(0,0)[lt]{\begin{minipage}{0.25938572\unitlength}\raggedright $m$ agents $U$ \end{minipage}}}%
    \put(0.24345999,0.0291763){\color[rgb]{0,0,0}\makebox(0,0)[lt]{\begin{minipage}{0.25033735\unitlength}\raggedright $m$ agents $V$\end{minipage}}}%
    \put(0,0){\includegraphics[width=\unitlength,page=3]{proof2-pic.pdf}}%
    \put(0.50957229,0.10209541){\color[rgb]{0,0,0}\makebox(0,0)[lt]{\begin{minipage}{0.04825782\unitlength}\raggedright $1$\end{minipage}}}%
    \put(0,0){\includegraphics[width=\unitlength,page=4]{proof2-pic.pdf}}%
  \end{picture}%
\endgroup%